\documentclass[preprint,11pt]{elsarticle}
\usepackage{amsmath}
\usepackage{amssymb}
\usepackage{latexsym}
\usepackage{graphicx}
\usepackage{bm}
\usepackage{nicefrac}
\usepackage[usenames]{color}

\input amssym.def
\newsymbol\rtimes 226F
\newfont{\nset}{msbm10}

\newtheorem{theo}{Theorem}[section]
\newtheorem{theorem}[theo]{Theorem}

\newtheorem{lemma}[theo]{Lemma}

\newtheorem{corollary}[theo]{Corollary}

\journal{Theoretical Computer Science}

\begin{document}

\begin{frontmatter}

\title{Independence number and the number of maximum independent sets in  pseudofractal scale-free web and Sierpi\'nski gasket}

\author[lable1,label2]{Liren Shan}

\author[lable1,label2]{Huan Li}

\author[lable1,label2]{Zhongzhi Zhang}
\ead{zhangzz@fudan.edu.cn}

\address[lable1]{School of Computer Science, Fudan
University, Shanghai 200433, China}
\address[label2]{Shanghai Key Laboratory of Intelligent Information
Processing, Fudan University, Shanghai 200433, China}

\begin{abstract}
As a fundamental subject of theoretical computer science, the maximum independent set (MIS) problem not only is of purely theoretical interest, but also has found wide applications in various fields. However, for a general graph determining the size of a MIS is NP-hard, and exact computation of the number of all MISs is even more difficult. It is thus of significant interest to seek special graphs for which the MIS problem can be exactly solved. In this paper, we address the MIS problem in the pseudofractal scale-free web and the Sierpi\'nski gasket, which have the same number of vertices and  edges. For both graphs, we determine exactly the independence number and the number of all possible MISs. The independence number of the pseudofractal scale-free web is as twice as the one of the Sierpi\'nski gasket. Moreover, the pseudofractal scale-free web has a unique MIS, while the number of MISs in the Sierpi\'nski gasket grows exponentially with the number of vertices. %We argue that the heterogeneous structure is responsible for the distinction of the size and number of MISs between the two studied graphs, which in turn shows that scale-free property has a strong effect on the MIS problem and its applications in power-law graphs.
\end{abstract}

\begin{keyword}
Maximum independent  set\sep Independence number\sep  Minimum vertex cover\sep Scale-free network\sep Sierpi\'nski gasket\sep Complex network
\end{keyword}

\end{frontmatter}

\section{Introduction}

An independent set of a graph $\mathcal{G}$ with vertex set $\mathcal{V}$ is a subset $\mathcal{I}$ of $\mathcal{V}$, such that each pair of vertices in $\mathcal{I}$ is not adjacent in $\mathcal{G}$. A maximal independent set is an independent set that is not a subset of any other independent set. A largest maximal independent set is called a maximum independent set (MIS). In other words, a MIS is an independent set that has the largest size or cardinality. The cardinality of a MIS is referred to as the independence number of graph $\mathcal{G}$. A graph $\mathcal{G}$ is called a unique independence graph if it has a
unique MIS~\cite{HoSt85}. The MIS problem has a close connection with many other fundamental graph problems~\cite{Ro86,BeFu94,HaRa97}. For instance, the MIS problem in a graph  is equivalent to  the minimum vertex cover problem~\cite{Ka72} in the same graph, as well as the maximum clique problem in its complement graph~\cite{PaXu94}. In addition, the MIS problem is also closely related to graph coloring, maximum common induced subgraphs, and maximum common edge subgraphs~\cite{LiLuYaXiWe15}.

In addition to its intrinsic theoretical interest, the MIS problem has found important applications in a large variety of areas, such as coding theory~\cite{BuPaSeShSt02}, collusion detection in voting pools~\cite{ArFaDoSiKo11}, scheduling in wireless networks~\cite{JoLiRySh16}. For example, it was shown  in~\cite{BuPaSeShSt02} that the problem of finding the largest error correcting codes  can be reduced to the MIS problem on a graph. In~\cite{ArFaDoSiKo11} the problem of collusion detection was framed as identifying maximum independent sets. Moreover, finding a maximal weighted independent set in a wireless network is connected with the problem of organizing the vertices of the network in a hierarchical way~\cite{Ba01}. Finally, the MIS problem also has numerous applications in mining of graph data~\cite{LiLuYaXiWe15,ChLiZh17}.

In view of the theoretical and practical relevance, in the past decades the MIS problem has received much attention from different disciplines, e.g., theoretical computer science~\cite{MuPa02,XiNa13,AgHaLo13,HoKlLiLiPoWa15,LoMoRi15,ChEn16} and discrete mathematics~\cite{CeSt13,XiN16,Mo17}. It is well-known that solving the MIS problem of a generic graph is computationally difficult. Finding a MIS of a graph is a classic NP-hard problem~\cite{Ro86,HaRa97}, while enumerating all MISs in a graph is even \#P-complete~\cite{Va79TCS,Va79SiamJComput}. Due to the hardness of the MIS problem, exact algorithms for finding a MIS in a general graph take exponential time~\cite{FoKr10,TaTr77,BeImSa07}, which is infeasible for moderately
sized graphs. For practical applications, many local or heuristic algorithms were proposed to solve the MIS problem for those massive and intractable graphs~\cite{AnReWe12,DaLaSaScStWe16, LaSaScStWe17}.

Comprehensive empirical study~\cite{Ne03} has unveiled that large real networks are typically scale-free~\cite{BaAl99}, with their vertex degree following a power-law distribution $P(k) \sim k^{-\gamma}$. This nontrivial  heterogeneous structure has a strong effect on various topological and combinatorial aspects of a graph, such as average distances~\cite{ChLu02}, maximum matchings~\cite{LiSlBa11,ZhWu15}, and dominating sets~\cite{NaAk12, GaHaK15, ShLiZh17}. Although there have been concerted efforts to understanding the MIS problem in general, there has been
significantly less work focused on the MIS problem for power-law graphs~\cite{FePaPa08}. In particular, exact result about the independence number and the number of all MISs in a power-law graph is still lacking, despite the fact that exact result is helpful for testing heuristic algorithms. Moreover, the influence of scale-free behavior on the MIS problem is not well understood, although it is suggested to play an important role in the MIS problem.

%Despite the vast applications, solving the MDS problem of a graph is a challenge, because finding a MDS of a graph is NP-hard~\cite{HaHeSl98}. Over the past years, the MDS problem has attracted considerable attention from theoretical computer science~\cite{FoGrPySt08, HeIs12,dadedeMa14, GaHaK15, CoLeLi15}, discrete and combinatorial mathematics~\cite{MaTa96, KaLiMaNo14, HoKaNa10}, as well as statistical physics~\cite{ZhHaZh15}, and continues to be an active object of research~\cite{LiZhShXu16, GoHeKr16}. Extensive empirical research~\cite{Ne03} uncovered that most real networks exhibit the prominent scale-free behavior~\cite{BaAl99}, with the  degree of their vertices following a power-law distribution $P(k) \sim k^{-\alpha}$. Although an increasing number of studies have been focused on MDS problem, related works about MDS problem in scale-free networks are very rare~\cite{NaAk12, MoDeCzSzSzKo14}. In particular, exact results for domination number and the number of minimum dominating sets in a scale-free network are still lacking.

The ubiquity of power-law phenomenon makes it interesting to uncover the dependence of MISs on the scale-free feature, which is helpful for understanding the applications of MIS problem.  In this paper, we study the independence number and the number of maximum independent sets in a scale-free graph, called pseudofractal scale-free web~\cite{DoGoMe02,ZhQiZhXiGu09}, and the Sierpi\'nski gasket. Both networks are deterministic and have the same number of vertices and edges. Note that since determining the independence number and counting all maximum independent sets in a general graph are formidable, we choose these two exactly tractable graphs. This is a fundamental route of research for NP-hard and  \#P-complete problems. For example, Lov{\'a}sz~\cite{LoPl86} pointed out that it is of great interest to find specific graphs for which the matching problem can be exactly solved, since the problem in general graphs is NP-hard.

%Due to the ubiquity of scale-free phenomenon in realistic networks, unveiling the behavior of minimum dominating sets with respect to power-law degree distribution is important for better understanding the applications of MDS problem in real-life scale-free networks. On the other hand, determining the domination number and enumerating all minimum dominating sets in a generic network are formidable~\cite{HaHeSl98}, it is thus of great interest to find specific scale-free networks for which the MDS problem can be exactly solved~\cite{LoPl86}.

By using an analytic technique based on a decimation procedure~\cite{KnVa86}, we find the exact independence number and the number of all possible maximum independent sets for both studied graphs. The independence number of the pseudofractal scale-free web is as twice as the one associated with the Sierpi\'nski gasket. In addition to this difference, there is a unique maximum independent set in the pseudofractal scale-free web, while the number of all maximum independent sets in the Sierpi\'nski gasket increases as an exponential function of the number of vertices. %We show that the architecture dissimilarity between the two studied networks is responsible for the difference of maximum independent sets in these two nontrivial graphs. The pseudofractal scale-free web is heterogeneous, in sharp contrast to the Sierpi\'nski gasket that is homogeneous.

%In this paper, we focus on the domination number and the number of minimum dominating sets in a scale-free network, called pseudofractal scale-free web~\cite{DoGoMe02,ZhQiZhXiGu09} and the Sierpi\'nski gasket with the same number of vertices and edges. For both networks, we determine the exact domination number and the number of all different minimum dominating sets. The domination number of the pseudofractal scale-free web is only half of that of the Sierpi\'nski gasket. In addition, in the pseudofractal scale-free web, there is a unique minimum dominating set, while in Sierpi\'nski gasket the number of all different minimum dominating sets grows exponentially with the number of vertices in the graph. We show the root of the difference of minimum dominating sets between studied networks rests with their architecture dissimilarity, with the pseudofractal scale-free web being heterogeneous while the Sierpi\'nski gasket being homogeneous.

\section{Independence number and the number of maximum independent sets in pseudofractal scale-free web}

In this section, we study the independence number in the pseudofractal scale-free web, and demonstrate that its maximum independent set is unique.

\subsection{Network construction and properties}

The pseudofractal scale-free web~\cite{DoGoMe02,ZhQiZhXiGu09}  is constructed in an iterative way. Let  $\mathcal{G}_n$, $n\geq 1$, denote the network after $n$ iterations.  When $n=1$, $\mathcal{G}_1$ is a triangle. For $n > 1$,  $\mathcal{G}_n$ is obtained by adding, for every edge $(u,v)$ in $\mathcal{G}_{n-1}$ a new vertex connected to $u$ and $v$. Figure~\ref{network} illustrates the networks for the first several iterations. By construction,  the total number of edges in  $\mathcal{G}_n$ is $E_n=3^{n}$.

%%%%%%%%%%%%%%%%%%%%%%%%%%%%%%%%%%%%%%%%%%%%%%%%%%%%%%%%%
% Figure  1
%%%%%%%%%%%%%%%%%%%%%%%%%%%%%%%%%%%%%%%%%%%%%%%%%%%%%%%%%%
\begin{figure}
\begin{center}
\includegraphics{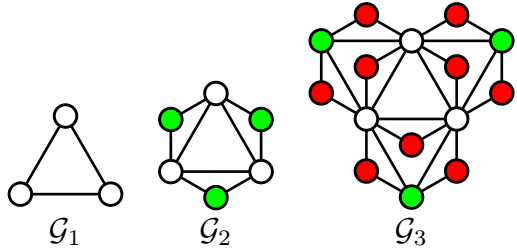} %[width=.60\linewidth,trim=60 40 60 40]
\end{center}
\caption[kurzform]{The first three iterations of the scale-free graph.} \label{network}
\end{figure}
%%%%%%%%%%%%%%%%%%%%%%%%%%%%%%%%%%%%%%%%%%%%%%%%%%%%%%%%%%

The  network displays the striking properties observed in most real-life networks. First, it is scale-free, since the degree of its vertices obeys a power law  distribution $P(k)\sim k^{-\ln 3/ \ln 2}$~\cite{DoGoMe02}, implying  that the probability of a vertex chosen randomly having degree $k$ is approximately $k^{- \ln 3/ \ln 2}$. Moreover, it is small-world, with its average distance growing logarithmically with the number of vertices~\cite{DoGoMe02,ZhZhCh07}. Finally, it is highly clustered, with its average clustering coefficient converging to  $\frac{4}{5}$.

Of particular interest is the self-similarity of network $\mathcal{G}_n$, which is another ubiquitous property  of real networks~\cite{SoHaMa05}.  For  $\mathcal{G}_n$, the three vertices generated at  $n=1$ have the highest degree, which are called hub vertices, and are denoted by $A_{n}$, $B_{n}$, and $C_{n}$, respectively. The self-similar feature of the network can be seen from another construction approach~\cite{ZhZhCh07}. Given  the network $\mathcal{G}_{n}$, $\mathcal{G}_{n+1}$ can be obtained by joining three copies of $\mathcal{G}_{n}$ at their hub vertices, see Fig.~\ref{mergeF}. Let $\mathcal{G}_{n}^{(\theta)}$, $\theta=1,2,3$, be three replicas of $\mathcal{G}_{n}$, and denote  the three hub vertices of $\mathcal{G}_{n}^{(\theta)}$  by $A_{n}^{(\theta)}$, $B_{n}^{(\theta)}$,
and $C_{n}^{(\theta)}$, respectively. Then, $\mathcal{G}_{n+1}$ can be obtained by merging $\mathcal{G}_{n}^{(\theta)}$, $\theta=1,2,3$,  with $A_{n}^{(1)}$
(resp. $C_{n}^{(1)}$, $A_{n}^{(2)}$) and $B_{n}^{(3)}$ (resp. $B_{n}^{(2)}$,
$C_{n}^{(3)}$) being identified as the hub vertex $A_{n+1}$ (resp.
$B_{n+1}$, $C_{n+1}$) in $\mathcal{G}_{n+1}$.

%%%%%%%%%%%%%%%%%%%%%%%%%%%%%%%%%%%%%%%%%%%%%%%%%%%%%%%%%
% Figure  2
%%%%%%%%%%%%%%%%%%%%%%%%%%%%%%%%%%%%%%%%%%%%%%%%%%%%%%%%%%
\begin{figure}
\begin{center}
\includegraphics[width=.60\linewidth]
{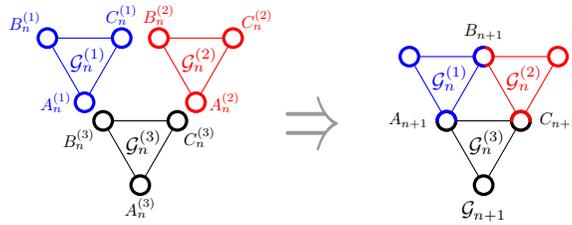}
\caption{Alternative construction of the scale-free network.} \label{mergeF}
\end{center}
\end{figure}
%%%%%%%%%%%%%%%%%%%%%%%%%%%%%%%%%%%%%%%%%%%%%%%%%%%%%%%%%%

Let $N_n$ stand for  the total number of vertices in $\mathcal{G}_{n}$. By the second construction of the network, $N_n$ satisfies  relation $N_{n+1}=3N_n-3$, which together with the initial value $N_1=3$, is solved to give $N_n=(3^{n}+3)/2$.

%\section{The Construction of Pseudofractal Scale-free Network}
%We use two different ways to define pseudofractal scale-free network.
%\begin{definition}
%The pseudofractal scale-free network, denoted by $G_n,n\geq1$ after n generation evolution, is constructed as follows.
%\begin{description}
%\item[(i)] for $n=1,G_1$ has three vertices with edges between them.
%\item[(ii)]for $n>1$, $G_n$ is obtained by adding a new vertex for every edge of $G_{n-1}$, which is connected to two ends of this edge.
%\end{description}
%\end{definition}
%\begin{figure}[htbp]
%\centering
%\includegraphics{graph5.eps}
%\caption{The first three generations of pseudofractal scale-free network. New vertices in every generation are colored by yellow, red and white respectively.}
%\end{figure}
%In this iterative construction, there are three vertices exist at the first generation, which are called hub vertices of network.
%According to the self-similarity of this network, it can be defined in another way.
%\begin{definition}
%Given the network $G_n$ of generation g, $G_{n+1}$ is obtained by performing the following three operations.
%\begin{description}
%\item[(i)] Taking three copies of $G_n$, denoted as $G_n^k, k=1,2,3$.
%\item[(ii)]Introducing a triangle as a center. For each edge of the center, two hub vertices of copies in first step are attached to the two ends of this edge.
%\item[(iii)] Identifying the two vertices attached above as one hub vertex of new network $G_{n+1}$.
%\end{description}
%\end{definition}

\subsection{Independence number and the number of maximum independent  sets}

%According to the construction of pseudofractal scale-free network, there are three hub vertices on these networks. Hence, we use $\alpha_n^k$ for $k = 0,1,2,3$ to denote the minimum number of domination sets of $G_n$ with $k$ hub vertices respectively. And $\alpha_n$ is used to denote the domination number of $G_n$. By the definition, we can easily to get the lemma as follow:

%%%%%%%%%%%%%%%%%%%%%%%%%%%%%%%%%%%%%%%%%%%%%%%%%%%%%%%%%
% Figure 3
%%%%%%%%%%%%%%%%%%%%%%%%%%%%%%%%%%%%%%%%%%%%%%%%%%%%%%%%%%
%\begin{figure}
%\begin{center}
%\includegraphics[width=0.70\linewidth]%,trim=0 80 0 500]
%{DemoDomiSet.eps}
%\end{center}
%\caption[kurzform]{\label{Demo01} Illustrations for the definitions of
%$\Theta_n^k$, $k = 0,1,2,3$, only showing the three hub vertices. (a), (b), (c) and (d) correspond to a dominating set belonging to $\Theta_n^0$, $\Theta_n^1$, $\Theta_n^2$, and $\Theta_n^3$, respectively. Every filled circle denotes a hub  in the dominating set, while each empty circle represents a hub not in the dominating set.}
%\end{figure}
%%%%%%%%%%%%%%%%%%%%%%%%%%%%%%%%%%%%%%%%%%%%%%%%%%%%%%%%%%

Let $\alpha_n$ denote the  independence number of network $\mathcal{G}_n$. To determine $\alpha_n$, we introduce some intermediate quantities. Since the three hub vertices in $\mathcal{G}_n$ are connected to each other, any independent  set of $\mathcal{G}_n$ contains at most one hub vertex.  We  classify all independent  sets of $\mathcal{G}_n$ into two subsets $\Omega_n^0$ and $\Omega_n^1$. $\Omega_n^0$ represents those independent  sets with no hub vertex, while  $\Omega_n^1$ denotes the remaining  independent  sets, with each having exactly one hub vertex. Let $\Theta_n^k$, $k = 0,1$, be the subset of $\Omega_n^k$, where each independent set has the largest cardinality (number of vertices), denoted by $\alpha_n^k$. By definition,  the independence number of network $\mathcal{G}_n$, $n\geq1$, is $\alpha_n = \max\{\alpha_n^0,\alpha_n^1\}$.  %we have the following lemma.%Figure~\ref{Demo01} illustrates the definitions for $\Theta_n^k$, $k = 0,1$.

%\begin{lemma}\label{Ind01}
%The independence number of network $\mathcal{G}_n$, $n\geq1$, is $\alpha_n = \max\{\alpha_n^0,\alpha_n^1\}$.
%\end{lemma}
%After reducing the problem of determining $\alpha_n$ to computing $\alpha_n^k$, $k = 0,1$, we next evaluate $\alpha_n^k$ by
The two quantities $\alpha_n^0$ and $\alpha_n^1$ can be evaluated by using  the self-similar structure of the network. %there are some relations between the notations we give above.
\begin{lemma}
For two successive generation networks $\mathcal{G}_n$ and $\mathcal{G}_{n+1}$, $n\geq1$,%\par
\begin{equation}\label{Ind02}
\alpha_{n+1}^0 = \max\{3\alpha_n^0,2\alpha_n^0+\alpha_n^1,\alpha_n^0 +2\alpha_n^1,3\alpha_n^1\},
\end{equation}
\begin{equation}\label{Ind03}
\alpha_{n+1}^1 = \max\{2\alpha_n^1+\alpha_n^0-1,3\alpha_n^1-1\}.
\end{equation}
\end{lemma}

\begin{proof}
By definition, $\alpha_{n+1}^k$, $k = 0,1$, is the cardinality of an independent set in $\Theta_{n+1}^k$. Below, we will show that both $\Theta_{n+1}^0$ and $\Theta_{n+1}^1$ can be constructed iteratively from $\Theta_n^0$ and $\Theta_n^1$. Then, $\alpha_{n+1}^0$ and $\alpha_{n+1}^1$  can be obtained from $\alpha_n^0$ and $\alpha_n^1$. We now establish the recursive relations for  $\alpha_n^0$ and $\alpha_n^1$.

We first prove graphically Eq.~\eqref{Ind02} .

Notice that $\mathcal{G}_{n+1}$ consists of three copies of $\mathcal{G}_n$, $\mathcal{G}_{n}^{(\theta)}$, $\theta=1,2,3$. By definition, for any independent  set $\chi$ in $\Theta_{n+1}^0$, the three hub vertices of $\mathcal{G}_{n+1}$ do not belong to $\chi$,  implying  that the corresponding six identified hub vertices of $\mathcal{G}_{n}^{(\theta)}$,  $\theta=1,2,3$, are not in $\chi$, see Fig.~\ref{mergeF}. Therefore, we can construct set $\chi$ from $\Theta_n^0$ and $\Theta_n^1$ by considering whether the hub vertices of $\mathcal{G}_{n}^{(\theta)}$, $\theta=1,2,3$, are in  $\chi$ or not.  Fig.~\ref{Theta0} illustrates all possible configurations of independent  sets $\Omega_{n+1}^0$ that include $\Theta_{n+1}^0$ as subsets. From Fig.~\ref{Theta0}, we obtain
\begin{equation*}
\alpha_{n+1}^0 = \max\{3\alpha_n^0,2\alpha_n^0+\alpha_n^1,\alpha_n^0 +2\alpha_n^1,3\alpha_n^1\}.
\end{equation*}

Similarly we can  prove Eq.~\eqref{Ind03}, the graphical representation of which is shown in Fig.~\ref{Theta1}.
\end{proof}
%we don't know whether they are in the domination set. Therefore, We have to consider all configurations in each situation. there are still three hub vertices of these copies.

%As we mentioned above, there are four different situations for configuration of hub vertices. However, except the new hub vertices of $G_{n+1}$, there are still three hub vertices of these copies. For these three vertices, we don't know whether they are in the domination set. Therefore, We have to consider all configurations in each situation.\par

%First, all three hub vertices of $G_{n+1}$ are not in domination set. Considering symmetry and rotation, there are four intrinsic different configurations in showed in the Figure1.

%Similarly, we can enumerate all configurations of other three situations.

%%%%%%%%%%%%%%%%%%%%%%%%%%%%%%%%%%%%%%%%%%%%%%%%%%%%%%%%
% Figure  4
%%%%%%%%%%%%%%%%%%%%%%%%%%%%%%%%%%%%%%%%%%%%%%%%%%%%%%%%%
\begin{figure}[htbp]
\centering
\includegraphics[width=0.90\linewidth]{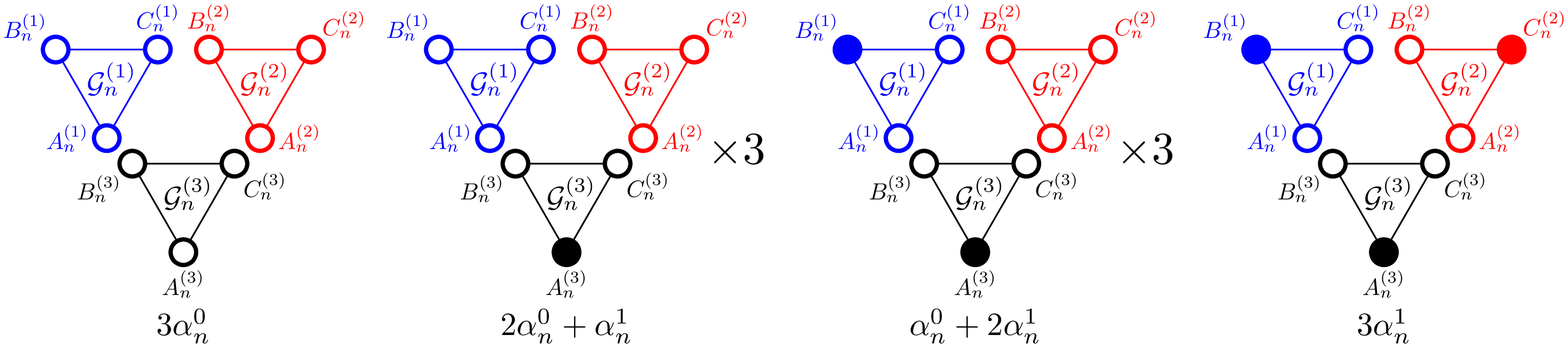}
\caption{\label{Theta0}Illustration of all possible configurations of independent sets $\Omega_{n+1}^0$ of $\mathcal{G}_{n+1}$, which contain $\Theta_{n+1}^0$. Only the hub vertices of $\mathcal{G}_{n}^{(\theta)}$, $\theta=1,2,3$, are shown. Filled vertices are in the independent sets, while open vertices are not.}
\end{figure}
%%%%%%%%%%%%%%%%%%%%%%%%%%%%%%%%%%%%%%%%%%%%%%%%%%%%%%%%%

%%%%%%%%%%%%%%%%%%%%%%%%%%%%%%%%%%%%%%%%%%%%%%%%%%%%%%%%
 %Figure  5
%%%%%%%%%%%%%%%%%%%%%%%%%%%%%%%%%%%%%%%%%%%%%%%%%%%%%%%%%
\begin{figure}[htbp]
\centering
\includegraphics[width=0.6\linewidth]{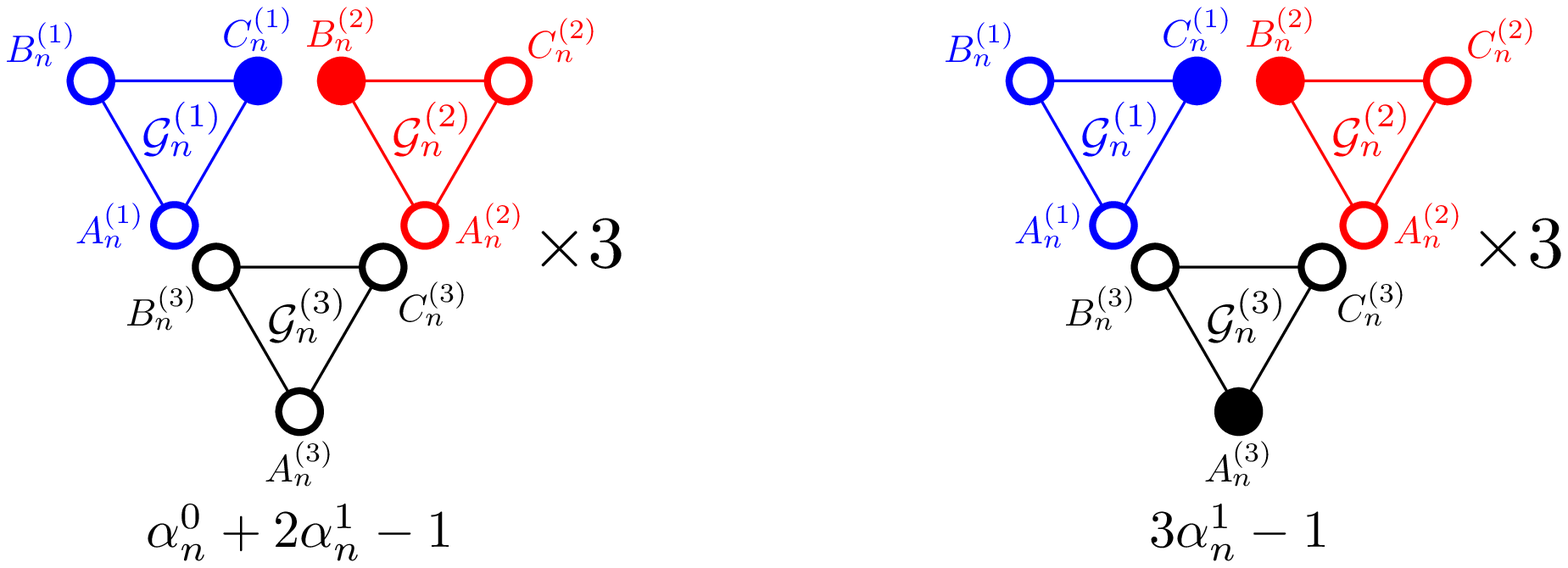}
\caption{\label{Theta1}Illustration of all possible configurations of  independent sets $\Omega_{n+1}^1$ of $\mathcal{G}_{n+1}$, which contain $\Theta_{n+1}^1$.}
\end{figure}
%%%%%%%%%%%%%%%%%%%%%%%%%%%%%%%%%%%%%%%%%%%%%%%%%%%%%%%%%

%%%%%%%%%%%%%%%%%%%%%%%%%%%%%%%%%%%%%%%%%%%%%%%%%%%%%%%%%
% Figure  6
%%%%%%%%%%%%%%%%%%%%%%%%%%%%%%%%%%%%%%%%%%%%%%%%%%%%%%%%%%
%\begin{figure}[htbp]
%\centering
%\includegraphics[width=1.00\linewidth]{SfTheta2.eps}
%\caption{\label{Theta2}Illustration of all possible configurations of dominating sets $\Omega_{n+1}^2$ of $\mathcal{G}_{n+1}$, which contain $\Theta_{n+1}^2$.}
%\end{figure}
%%%%%%%%%%%%%%%%%%%%%%%%%%%%%%%%%%%%%%%%%%%%%%%%%%%%%%%%%%

%%%%%%%%%%%%%%%%%%%%%%%%%%%%%%%%%%%%%%%%%%%%%%%%%%%%%%%%%
% Figure  7
%%%%%%%%%%%%%%%%%%%%%%%%%%%%%%%%%%%%%%%%%%%%%%%%%%%%%%%%%%
%\begin{figure}[htbp]
%\centering
%\includegraphics[width=0.60\linewidth]{SfTheta3.eps}
%\caption{\label{Theta3}Illustration of all possible configurations of dominating sets $\Omega_{n+1}^3$ of $\mathcal{G}_{n+1}$, which contain $\Theta_{n+1}^3$.}
%\end{figure}
%%%%%%%%%%%%%%%%%%%%%%%%%%%%%%%%%%%%%%%%%%%%%%%%%%%%%%%%%%

\begin{lemma}\label{Dom06}
For network $\mathcal{G}_n$, $n\geq 2$, $\alpha_n^1 < \alpha_n^0$.
\end{lemma}
\begin{proof}
We prove this lemma by mathematical induction on $n$. For $n=2$, we  obtain $\alpha_2^1=2$, $\alpha_2^0=3$ by hand. Thus, the basis step holds immediately. \par
Assume that the statement holds for $t$ ($t\geq 2$). Then, according to Eq.~\eqref{Ind02}, $\alpha_{t+1}^0 = \max\{3\alpha_t^0,2\alpha_t^0+\alpha_t^1,\alpha_t^0 +2\alpha_t^1,3\alpha_t^1\}$. By induction hypothesis, we have
\begin{equation}\label{Falpha01}
\alpha_{t+1}^0 = 3\alpha_t^0.
\end{equation}
In an analogous way, we  obtain  relation
\begin{equation}\label{Falpha02}
\alpha_{t+1}^1 = \alpha_t^0 + 2\alpha_t^1-1.
\end{equation}
By comparing Eqs.~\eqref{Falpha01} and~\eqref{Falpha02} and using the induction hypothesis $\alpha_t^1 < \alpha_t^0$, we have $\alpha_{t+1}^1 < \alpha_{t+1}^0$. Thus, the lemma is true for $t+1$.
%This completes the proof.
\end{proof}

%Using two lemmas above, we can easily calculate the domination number of pseudofractal scale-free network.
\begin{theorem}\label{SFIndN}
The independence number of network $\mathcal{G}_n$, $n\geq 1$, is
\begin{equation}\label{Falpha04x}
\alpha_n = 3^{n-1}\,.
\end{equation}
\end{theorem}
\begin{proof}
Lemma~\ref{Dom06}  indicates that any maximum independent vertex set of $\mathcal{G}_n$ contains no hub vertices. By Eq.~\eqref{Falpha01}, we obtain
\begin{equation}\label{Falpha05}
\alpha_{n+1} = \alpha^0_{n+1} = 3\alpha^0_n = 3\alpha_n \,.
\end{equation}
Considering the initial condition $\alpha_{1}=1$, the above equation is solved to give the result.
\end{proof}

%Similarly, we can get exact solutions of $\alpha_n^i$ for $i = 0,1,2$ by the recursive equations.
%Theorem~\ref{SFIndN}, particularly Eq.~\eqref{Falpha05}, indicates that any maximum independent vertex set of $\mathcal{G}_n$ contains no hub vertices.

\begin{corollary}
The largest number of vertices in an independent vertex set of $\mathcal{G}_n$, $n\geq2$, which contains exactly $1$ hub vertex, is
\begin{equation}\label{Falpha07}
\alpha_n^1 = 3^{n-1} - 2^{n-1} + 1.
\end{equation}
\end{corollary}
\begin{proof}
By Eqs.~\eqref{Falpha04x} and~\eqref{Falpha05}, we derive $\alpha_n^0=\alpha_n=3^{n-1}$. Using  Eq.~\eqref{Falpha02}, we obtain the following recursive equation for $\alpha_n^1$:
\begin{equation}\label{Falpha06}
\alpha_{n+1}^1 = 2\alpha_n^1 + 3^{n-1}-1,
%\alpha_{n+1}^2 -\frac{3^{n-1}+1}{2} &=& 2\left(\alpha_n^2 - \frac{3^{n-2}+1}{2}\right) \nonumber
\end{equation}
which together with the boundary condition $\alpha_2^1 = 2$ is solved to yield Eq.~\eqref{Falpha07}.
%This completes the proof.
\end{proof}

\begin{theorem}\label{ThereoreGn}
For network $\mathcal{G}_n$,  $n \geq 2$, there is a unique maximum independent  set.
\end{theorem}

\begin{proof}
Eq.~\eqref{Falpha05} and Fig.~\ref{Theta0} mean that for $n \geq 2$ any maximum independent set of $\mathcal{G}_{n+1}$ is actually the union of maximum independent sets, $\Theta_n^0$, of the three copies of $\mathcal{G}_{n}$ (i.e. $\mathcal{G}_{n}^{(1)}$, $\mathcal{G}_{n}^{(2)}$, and $\mathcal{G}_{n}^{(3)}$) constituting $\mathcal{G}_{n+1}$. Thus, any maximum independent set of $\mathcal{G}_{n+1}$ is determined by those of $\mathcal{G}_{n}^{(1)}$, $\mathcal{G}_{n}^{(2)}$, and $\mathcal{G}_{n}^{(3)}$.  Because the maximum independent set of $\mathcal{G}_2$ is unique, there is a unique  maximum independent set for $\mathcal{G}_n$ for all $n \geq 2$. Furthermore,  the unique maximum independent set of $\mathcal{G}_n$, $n \geq 2$, is in fact the set of all vertices that are generated at the $(n-1)$-th iteration.
\end{proof}

Theorem~\ref{ThereoreGn} indicates that the pseudofractal scale-free web is a unique independence graph.

\section{Independence number and the number of maximum independent  sets in Sierpi\'nski gasket}

In this section, we consider the independence number and the number of maximum independent  sets in the Sierpi\'nski gasket, and compare the results with those of the pseudofractal scale-free web, with an aim to unveil the effect of network structure, in particular the scale-free property, on the independence  number and the number of maximum independent sets.

\subsection{Construction of Sierpi\'nski gasket}

The Sierpi\'nski gasket is also constructed iteratively. Let  $\mathcal{S}_n$, $n\geq 1$, represent the $n$-generation graph. For $n=1$, $\mathcal{S}_1$ is an equilateral triangle with three vertices and three edges. For $n =2$, perform a bisection of the three edges of  $\mathcal{S}_1$   forming four smaller replicas of the original equilateral triangle, and remove the central downward pointing  equilateral triangle to get $\mathcal{S}_2$. For $n>2$, $\mathcal{S}_n$ is obtained from $\mathcal{S}_{n-1}$ by performing the  %trisecting and removing
above two operations for each triangle in $\mathcal{S}_{n-1}$. Fig.~\ref{Sierp01} illustrates the first several iterations of the Sierpi\'nski gaskets $\mathcal{S}_n$ for $n=1,2,3$.

%%%%%%%%%%%%%%%%%%%%%%%%%%%%%%%%%%%%%%%%%%%%%%%%%%%%%%%%
%Figure  8
%%%%%%%%%%%%%%%%%%%%%%%%%%%%%%%%%%%%%%%%%%%%%%%%%%%%%%%%%
\begin{figure}
\begin{center}
\includegraphics[width=0.40\textwidth]{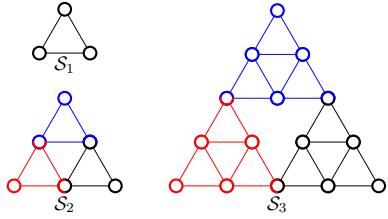} %[width=.60\linewidth,trim=60 40 60 40]
\end{center}
\caption[kurzform]{The first three generations of the
Sierpi\'nski gasket.} \label{Sierp01}
\end{figure}
%%%%%%%%%%%%%%%%%%%%%%%%%%%%%%%%%%%%%%%%%%%%%%%%%%%%%%%%%

Both the number of vertices and the number of edges in the Sierpi\'nski gasket $\mathcal{S}_n$ are the same as those for the scale-free network $\mathcal{G}_n$, which are equal to $N_n=(3^{n}+3)/2$ and $E_n=3^{n}$, respectively.

In contrast to the inhomogeneity of $\mathcal{G}_n$, the Sierpi\'nski gasket is homogeneous. The degree of vertices in $\mathcal{S}_n$ is equal to 4, except  the topmost vertex $A_n$, the leftmost vertex $B_n$, and the rightmost vertex $C_n$, the degree of which is 2. These three vertices with degree 2 are called  outmost vertices hereafter.

%%%%%%%%%%%%%%%%%%%%%%%%%%%%%%%%%%%%%%%%%%%%%%%%%%%%%%%%
%Figure  9
%%%%%%%%%%%%%%%%%%%%%%%%%%%%%%%%%%%%%%%%%%%%%%%%%%%%%%%%%
\begin{figure}
\begin{center}
\includegraphics[width=7.0cm]{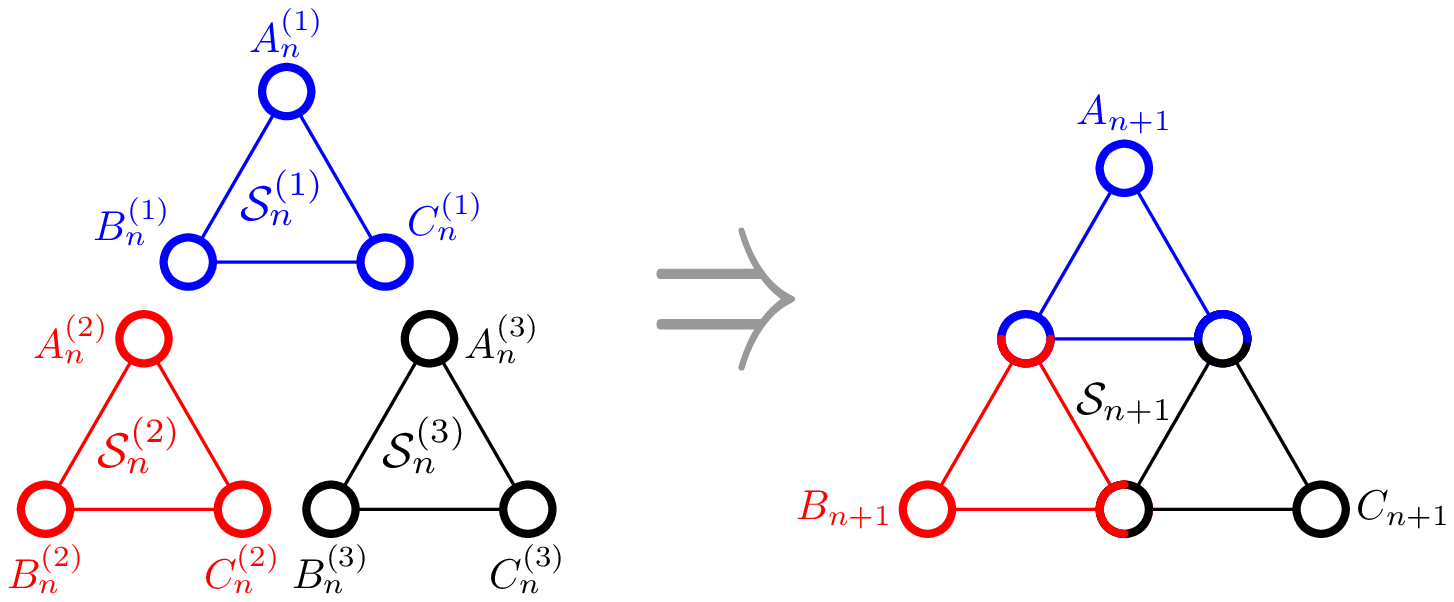}
\caption{Alternative construction of the Sierpi\'nski gasket.} \label{merge}
\end{center}
\end{figure}
%%%%%%%%%%%%%%%%%%%%%%%%%%%%%%%%%%%%%%%%%%%%%%%%%%%%%%%%%

Analogously  to the scale-free network $\mathcal{G}_n$, the Sierpi\'nski gasket also exhibits the self-similar property, which suggests an alternative construction way of the graph. Given the $n$th generation graph $\mathcal{S}_{n}$, the $(n+1)$th generation graph $\mathcal{S}_{n+1}$ can be obtained by amalgamating three copies of $\mathcal{S}_{n}$ at their outmost vertices, see Fig.~\ref{merge}. Let $\mathcal{S}_{n}^{(\theta)}$, $\theta=1,2,3$, represent three copies of $\mathcal{S}_{n}$. And denote the three outmost vertices of $\mathcal{S}_{n}^{(\theta)}$ by $A_{n}^{(\theta)}$, $B_{n}^{(\theta)}$, and $C_{n}^{(\theta)}$, respectively. Then, $\mathcal{S}_{n+1}$ can be obtained by coalescing $\mathcal{S}_{n}^{(\theta)}$, $\theta=1,2,3$, with $A_{n}^{(1)}$, $B_{n}^{(2)}$, and $C_{n}^{(3)}$ being the outmost vertices $A_{n+1}$, $B_{n+1}$, and $C_{n+1}$  of $\mathcal{S}_{n+1}$.

%First, we give the construction of Sierpi\'nski network.
%\begin{definition}
%By the self-similarity of network, the definition is given in a recursive way. The Sierpi\'nski network of generation $n$ is denoted by $S_n$.
%\begin{description}
%\item[(i)] When $n=1$, $S_1$ is a complete graph with three vertices. Specially, these three vertices are called corner vertices.
%\item[(ii)] The Sierpi\'nski network $S_{n+1}$ consists of three copies of $S_n$.
%\item[(iii)] Three copies are connected through the corner vertices as figure below. And three outer corner vertices of copies are new corner vertices of $S_{n+1}$.
%\end{description}
%\end{definition}

\subsection{Independence   number}

In this case without causing confusion,  we employ the same notation as those for $\mathcal{G}_n$ to study related quantities for the Sierpi\'nski gasket $\mathcal{S}_n$. Let $\alpha_n$ be the independence number of $\mathcal{S}_n$. Note that all independent  sets of $\mathcal{S}_n$ can be sorted into four types: $\Omega_n^0$, $\Omega_n^1$, $\Omega_n^2$, and $\Omega_n^3$, where $\Omega_n^k$,  $k = 0,1,2,3$,  stands for those independent sets, each of which includes exactly $k$ outmost vertices of $\mathcal{S}_n$. Let $\Theta_n^k$, $k = 0,1,2,3$, denote the subsets of $\Omega_n^k$, each independent set in which has the largest cardinality, denoted by $\alpha_n^k$.  Then, the independence  number of the Sierpi\'nski gasket $\mathcal{S}_n$, $n\geq1$, is $\alpha_n = \max\{\alpha_n^0,\alpha_n^1,\alpha_n^2,\alpha_n^3\}$.   Therefore,  to  determine $\alpha_n$ for $\mathcal{S}_n$, one can alternatively determine $\alpha_n^k$, $k = 0,1,2,3$, which can be solved by  establishing some relations between them, based on the self-similar architecture of the  Sierpi\'nski gasket.  %Figure~\ref{Demo01} illustrates the definitions for $\Theta_n^k$, $k = 0,1,2,3$.

%%%%%%%%%%%%%%%%%%%%%%%%%%%%%%%%%%%%%%%%%%%%%%%%%%%%%%%%%
% Figure 10
%%%%%%%%%%%%%%%%%%%%%%%%%%%%%%%%%%%%%%%%%%%%%%%%%%%%%%%%%%
%\begin{figure}
%\begin{center}
%\includegraphics[width=0.70\linewidth,trim=0 80 0 500]{DemoDomiSet.eps}
%\end{center}
%\caption[kurzform]{\label{Demo01} Illustrations for the definitions of
%$\Theta_n^k$, $k = 0,1,2,3$, only showing the three hub vertices. (a), (b), (c) and (d) correspond to a dominating set belonging to $\Theta_n^0$, $\Theta_n^1$, $\Theta_n^2$, and $\Theta_n^3$, respectively. Empty circle denotes a hub  in the dominating set, while filled circle represents a hub not in the dominating set.}
%\end{figure}
%%%%%%%%%%%%%%%%%%%%%%%%%%%%%%%%%%%%%%%%%%%%%%%%%%%%%%%%%%

%\begin{lemma}\label{leSGInd01}
%The independence  number of the Sierpi\'nski gasket $\mathcal{S}_n$, $n\geq1$, is $\alpha_n = \max\{\alpha_n^0,\alpha_n^1,\alpha_n^2,\alpha_n^3\}$.
%\end{lemma}

%According to the self-similar property of the Sierpi\'nski gasket, we can establish some relations between the quantities defined above.
\begin{lemma}
\label{leSGInd02}
For any integer $n \geq 3$, the following relations hold.
\begin{equation}\label{SGInd02}
\alpha_{n+1}^0 = \max\{3\alpha_n^0, \alpha_n^0+2\alpha_n^1-1, 2\alpha_n^1+ \alpha_n^2-2, 3\alpha_n^2-3\},
\end{equation}
\begin{equation}\label{SGInd03}
\alpha_{n+1}^1 = \max\{2\alpha_n^0+\alpha_n^1, \alpha_n^0+\alpha_n^1+\alpha_n^2-1, 3\alpha_n^1-1, 2\alpha_n^1+\alpha_n^3-2, \alpha_n^1+2\alpha_n^2-2, 2\alpha_n^2+\alpha_n^3-3\},
\end{equation}
\begin{equation}\label{SGInd04}
\alpha_{n+1}^2 = \max\{\alpha_n^0+2\alpha_n^1, \alpha_n^0+2\alpha_n^2-1, 2\alpha_n^1+\alpha_n^2-1,3\alpha_n^2-2, \alpha_n^1+\alpha_n^2+\alpha_n^3-2, \alpha_n^2+2\alpha_n^3-3\},
\end{equation}
\begin{equation}\label{SGInd05}
\alpha_{n+1}^3 = \max\{3\alpha_n^1, \alpha_n^1+2\alpha_n^2-1, 2\alpha_n^2+\alpha_n^3-2, 3\alpha_n^3-3\}.
\end{equation}
\end{lemma}
\begin{proof}
This lemma can be proved graphically. Figs.~\ref{SGTheta0}-\ref{SGTheta3} illustrate the graphical representations from Eq.~\eqref{SGInd02} to Eq.~\eqref{SGInd05}.
\end{proof}

%%%%%%%%%%%%%%%%%%%%%%%%%%%%%%%%%%%%%%%%%%%%%%%%%%%%%%%%
%Figure  10
%%%%%%%%%%%%%%%%%%%%%%%%%%%%%%%%%%%%%%%%%%%%%%%%%%%%%%%%%
\begin{figure}[htbp]
\centering
\includegraphics[width=0.70\textwidth]{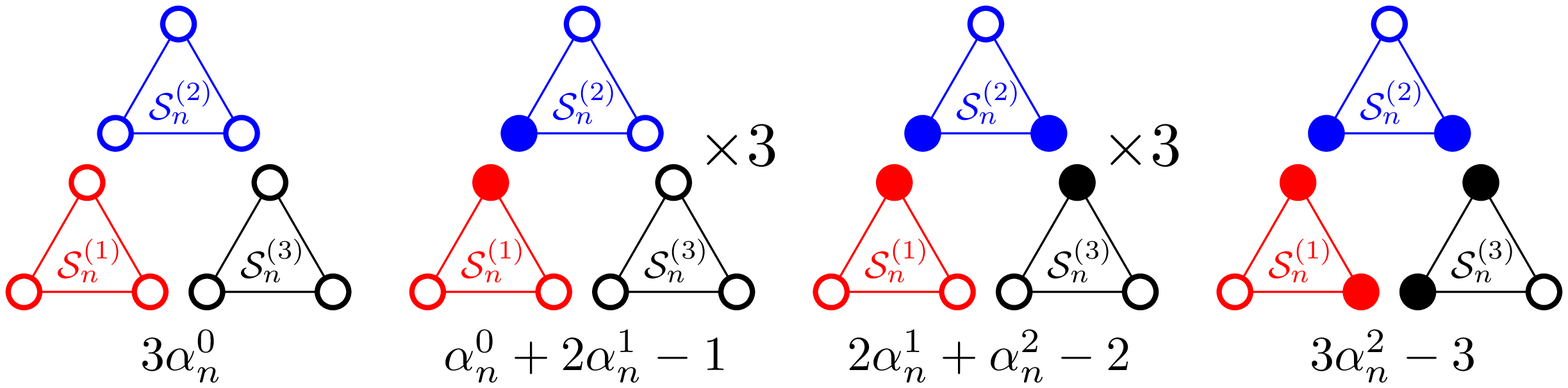}
\caption{\label{SGTheta0}Illustration of all possible configurations of independent sets $\Omega_{n+1}^0$ of $\mathcal{S}_{n+1}$, which contain $\Theta_{n+1}^0$. Only the outmost vertices of $\mathcal{S}_{n}^{(\theta)}$, $\theta=1,2,3$, are shown. Filled vertices are in the independent sets, while open vertices are not.}
\end{figure}
%%%%%%%%%%%%%%%%%%%%%%%%%%%%%%%%%%%%%%%%%%%%%%%%%%%%%%%%%

%%%%%%%%%%%%%%%%%%%%%%%%%%%%%%%%%%%%%%%%%%%%%%%%%%%%%%%%
% Figure  11
%%%%%%%%%%%%%%%%%%%%%%%%%%%%%%%%%%%%%%%%%%%%%%%%%%%%%%%%%
\begin{figure}[htbp]
\centering
\includegraphics[width=1.00\textwidth]{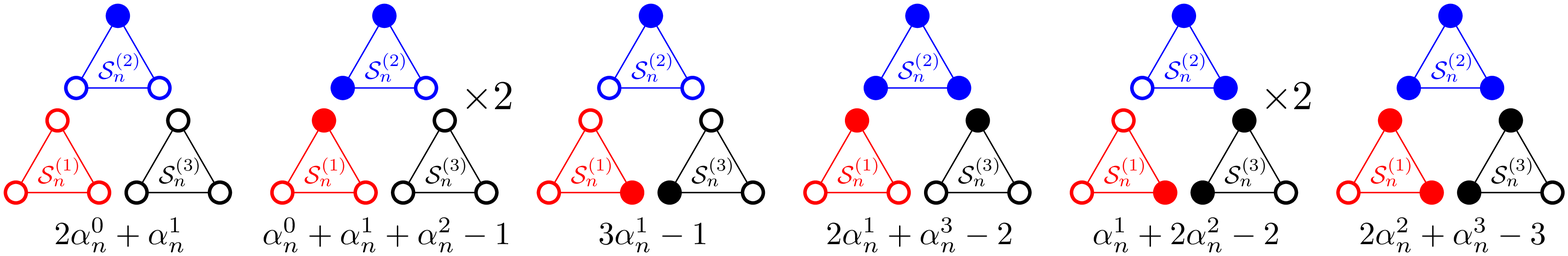}
\caption{\label{SGTheta1}Illustration of all possible configurations of  independent sets $\Omega_{n+1}^1$ of $\mathcal{S}_{n+1}$, which contain $\Theta_{n+1}^1$.  Note that here we only illustrate the  independent sets, each of which only includes $A_{n+1}$ but excludes   $B_{n+1}$, and $C_{n+1}$.  Since $A_{n+1}$, $B_{n+1}$, and $C_{n+1}$ are equivalent to each other, we omit other independent sets, including $B_{n+1}$ (resp. $C_{n+1}$) but excluding   $A_{n+1}$ and $C_{n+1}$ (resp.  $A_{n+1}$ and $B_{n+1}$) .}
\end{figure}
%%%%%%%%%%%%%%%%%%%%%%%%%%%%%%%%%%%%%%%%%%%%%%%%%%%%%%%%%

%%%%%%%%%%%%%%%%%%%%%%%%%%%%%%%%%%%%%%%%%%%%%%%%%%%%%%%%
% Figure  12
%%%%%%%%%%%%%%%%%%%%%%%%%%%%%%%%%%%%%%%%%%%%%%%%%%%%%%%%%
\begin{figure}[htbp]
\centering
\includegraphics[width=1.00\textwidth]{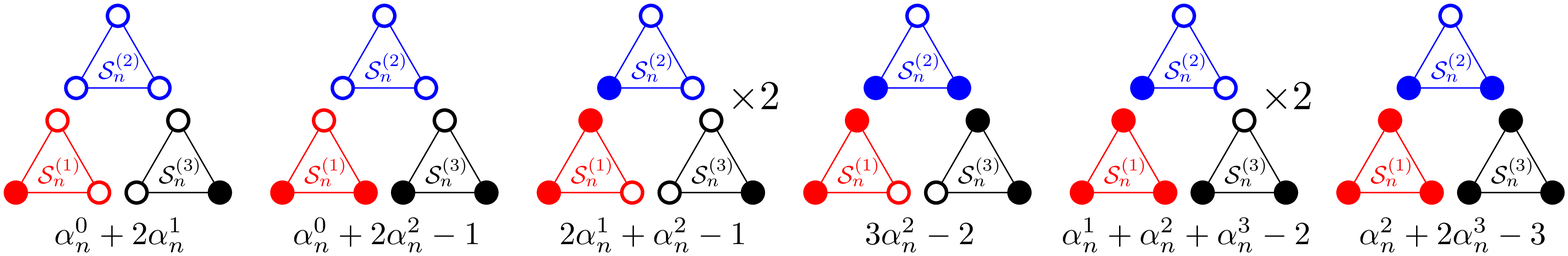}
\caption{\label{SGTheta2} Illustration of all possible configurations of  independent sets $\Omega_{n+1}^2$ of $\mathcal{S}_{n+1}$, which contain $\Theta_{n+1}^2$.  Note that here we only illustrate the  independent sets, each of which includes two outmost vertices $B_{n+1}$ and $C_{n+1}$, but excludes the outmost vertex  $A_{n+1}$.  Similarly, we can illustrate those independent sets, each  including  $A_{n+1}$ and $C_{n+1}$ (resp. $A_{n+1}$ and $B_{n+1}$), but excluding  $B_{n+1}$ (resp. $C_{n+1}$) .}
\end{figure}
%%%%%%%%%%%%%%%%%%%%%%%%%%%%%%%%%%%%%%%%%%%%%%%%%%%%%%%%%

%%%%%%%%%%%%%%%%%%%%%%%%%%%%%%%%%%%%%%%%%%%%%%%%%%%%%%%%
% Figure  13
%%%%%%%%%%%%%%%%%%%%%%%%%%%%%%%%%%%%%%%%%%%%%%%%%%%%%%%%%
\begin{figure}[htbp]
\centering
\includegraphics[width=0.70\textwidth]{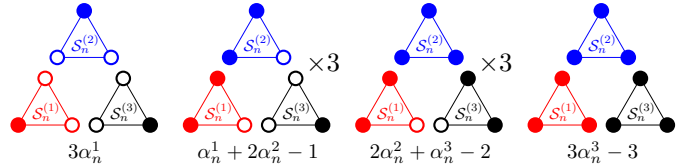}
\caption{\label{SGTheta3}Illustration of all possible configurations of  independent sets $\Omega_{n+1}^3$ of $\mathcal{S}_{n+1}$, which contain $\Theta_{n+1}^3$.}
\end{figure}
%%%%%%%%%%%%%%%%%%%%%%%%%%%%%%%%%%%%%%%%%%%%%%%%%%%%%%%%%

%\begin{lemma}
%when $n\geq 3$,
%\begin{eqnarray}
%\alpha_n^0 \leq \alpha_n^1 \leq \alpha_n^2 \leq \alpha_n^3 \\
%\phi_n^0 \leq \phi_n^1 \leq \phi_n^2 \\
%\xi_n^0 \leq \xi_n^1
%\end{eqnarray}
%\end{lemma}

%\begin{proof}
%For each corner vertex of Sierpi\'nski network, it has two neighbors. If a corner vertex is contained in the domination set, then there is an alternative way to dominate the graph that is changing this corner vertex to 5one of its neighbors. Hence, the total domination number shall not increase when there are less corner vertices in the domination set.
%\end{proof}

\begin{lemma}\label{leSGInd03}
For arbitrary $n \geq 2$, $\alpha_n^0+1 = \alpha_n^1 = \alpha_n^2 = \alpha_n^3-1$.
\end{lemma}
\begin{proof}
We prove this lemma by induction. \par For $n = 2$, it is easy to check that $\alpha_2^0 =1$, $\alpha_2^1 = 2$, $\alpha_2^2=2$, and $\alpha_2^3=3$. Thus, the result holds for $n = 2$. \par
Let us suppose that the statement is true for $ t$, $t \geq 2$.  For $t+1$, by
induction assumption and Lemma \ref{leSGInd02}, it is not difficult to check that the relation $\alpha_{t+1}^0+1 = \alpha_{t+1}^1 = \alpha_{t+1}^2 = \alpha_{t+1}^3-1$ is true. %Hence, this lemma is hold for $n \geq 3$.
\end{proof}

\begin{theorem}\label{alphan}
The independence  number of the Sierpi\'nski gasket $\mathcal{S}_n$, $n\geq2$, is $\alpha_n = \frac{3^{n-1}+3}{2}$.
\end{theorem}
\begin{proof}
According  to Lemmas \ref{leSGInd02}  and~\ref{leSGInd03}, we obtain  $\alpha_{n+1} = \alpha_{n+1}^3 = 3\alpha_n^3-3 = 3 \alpha_n-3$. Considering $\alpha_2 = 3$, it is obvious that $\alpha_n = \frac{3^{n-1}+3}{2}$ holds for all $n \geq 2$.
\end{proof}

Theorems~\ref{SFIndN} and Theorem~\ref{alphan} show that the independence number of the Sierpi\'nski gasket $\mathcal{S}_n$ is larger than the one corresponding to the  pseudofracal scale-free web $\mathcal{G}_n$, with the former  being  as  half as the latter for large $n$.

\begin{corollary}
The largest possible number of vertices in an independent vertex set of $\mathcal{S}_n$, $n\geq2$, which contains exactly $0$, $1$ and $2$ outmost vertices, is $\alpha_n^0 = \frac{3^{n-1}-1}{2}$, $\alpha_n^1 = \frac{3^{n-1}+1}{2}$, and $\alpha_n^2 = \frac{3^{n-1}+1}{2}$, respectively.
\end{corollary}
%\begin{equation}
%\alpha_n^0 = \frac{3^{n-1}-1}{2},
%\end{equation}
%\begin{equation}
%\alpha_n^1 = \frac{3^{n-1}+1}{2},
%\end{equation}
%and
%\begin{equation}
%\alpha_n^2 = \frac{3^{n-1}+1}{2},
%\end{equation}
%respectively.

\begin{proof}
Theorem~\ref{alphan} shows $\alpha_n^3=\alpha_n=\frac{3^{n-1}+3}{2}$. From Lemma~\ref{leSGInd03}, we obtain $\alpha_n^0 = \alpha_n^3 -2$ and $\alpha_n^1 = \alpha_n^2 = \alpha_n^3 - 1$. Then, the results are obtained immediately.
\end{proof}

\subsection{The number of maximum independent  sets}

In comparison with the scale-free network  $\mathcal{G}_{n}$ with a unique maximum independent  set, the number of maximum independent sets of  $\mathcal{S}_{n}$ increases exponentially with the number of vertices.

\begin{theorem}
For $n\geq2$, the number of  maximum independent  sets of the Sierpi\'nski gasket $\mathcal{S}_n$ is $2^{\frac{3^{n-2}-1}{2}}$.
\end{theorem}
\begin{proof}
Let $x_n$ denote the number of maximum independent sets of the Sierpi\'nski gasket $\mathcal{S}_{n}$. Let $y_n$ be the number of independent sets of  $\mathcal{S}_{n}$ with maximum number of vertices,   including only $A_{n}$ but excluding   $B_{n}$, and $C_{n}$. For  the initial condition $n=2$, we have $x_2=1$, $y_2=1$. For $n \geq 2$, we can prove that the two quantities $x_n$ and $y_n$  obey the following relations:
\begin{equation}\label{SGset02}
x_{n+1} = y_n^3+x_n^3,
\end{equation}
\begin{equation}\label{SGset03}
y_{n+1} = y_n^3+x_ny_n^2.
\end{equation}

We first prove Eq.~\eqref{SGset02}. By definition, $x_n$ is the number of different maximum independent  sets  for $\mathcal{S}_n$, each of which contains all the three outmost vertices of $\mathcal{S}_n$.
According to Lemma~\ref{leSGInd03} and Fig.~\ref{SGTheta3},  the two configurations in Fig.~\ref{SGTheta3} which maximize $\alpha^3_{n+1}$ (and thus $\alpha_{n+1}$) are when $S^{(\theta)}_n$, $\theta = 1, 2, 3$, contains exactly one outmost vertex and when it contains the three outmost vertices. Then, we can establish Eq.~\eqref{SGset02} by using the rotational symmetry of the Sierpi\'nski gasket.

Eq.~\eqref{SGset03} can be proved analogously by using  Lemma~\ref{leSGInd03} and Fig.~\ref{SGTheta1}.

Since $x_2=1$ and $y_2=1$,  Eqs.~\eqref{SGset02} and~\eqref{SGset03} show that  $x_n = y_n$ for all $n \geq 2$.  Then, we obtain a recursion relation for $x_n$ as $x_{n+1} = 2 x_n^3$, which together with the initial value  $x_2=1$ is solved to yield $x_n=2^{\frac{3^{n-2}-1}{2}}$.
\end{proof}

\section*{Acknowledgements}

This work is supported by the National Natural Science Foundation of China under Grant No. 11275049. %\textcolor[rgb]{0.00,0.00,1.00}{The authors are grateful to the anonymous reviewers for their valuable comments and suggestions, which have led to improvement of this paper.}

%\section*{References}


%\bibliography{cover}

\end{document}